\documentclass[11pt]{article}
\usepackage[utf8]{inputenc}
\usepackage{hyperref}
\usepackage{amsmath,amsthm,amssymb,caption,subcaption,nicefrac}
\usepackage{thmtools}
\usepackage{fullpage}
\usepackage{bm}
\usepackage[table]{xcolor}
\usepackage{multirow}
\usepackage{arydshln}
\usepackage{csquotes}
\usepackage{anyfontsize}
\usepackage[shortlabels]{enumitem}
\usepackage{authblk}

\usepackage{natbib}
\setcitestyle{authoryear,open={(},close={)}} 

\usepackage{algorithm}
\usepackage{algorithmic}
\usepackage[algo2e,ruled,noend]{algorithm2e}

\usepackage{tikz}

\usepackage{hyperref}
\hypersetup{
	colorlinks=true,
	linkcolor=blue,
	citecolor=blue,
	filecolor=blue,
	urlcolor=blue
}

\usepackage[nameinlink]{cleveref} 
\crefname{algocf}{algorithm}{algorithms}
\Crefname{algocf}{Algorithm}{Algorithms}
\makeatletter
\AddToHook{cmd/appendix/before}{\def\cref@section@alias{appendix}\def\cref@subsection@alias{appendix}}
\makeatother

\newtheorem{theorem}{Theorem}
\newtheorem{proposition}[theorem]{Proposition}
\newtheorem{lemma}[theorem]{Lemma}
\newtheorem{corollary}[theorem]{Corollary}

\newtheorem{definition}[theorem]{Definition}

\newtheorem{remark}[theorem]{Remark}

\usepackage{hyperref}

\usepackage{bm}

\usepackage{booktabs}
\usepackage{algorithm}
\usepackage{algorithmic}

\usepackage{nicefrac,xfrac}

\definecolor{Gred}{RGB}{219, 50, 54}
\definecolor{Ggreen}{RGB}{60, 186, 84}
\definecolor{Gblue}{RGB}{72, 133, 237}
\definecolor{Gyellow}{RGB}{247, 178, 16}
\definecolor{ToCgreen}{RGB}{0, 128, 0}
\definecolor{myGold}{RGB}{231,141,20}
\definecolor{myBlue}{rgb}{0.19,0.41,.65}
\definecolor{myPurple}{RGB}{175,0,124}
\definecolor{Gmagenta}{RGB}{255, 0, 255}

\usepackage{enumitem}
\setitemize[1]{leftmargin=4mm,itemsep=1pt}
\setenumerate[1]{leftmargin=4mm,itemsep=1pt}

\DeclareMathOperator{\score}{score}
\DeclareMathOperator{\Potential}{\mathtt{val}}
\newcommand{\ring}{\mathbb{K}}
\newcommand{\zero}{\mathbf{0}}
\newcommand{\one}{\mathbf{1}}

\newcommand{\cA}{\mathcal{A}}

\newcommand{\eps}{\varepsilon}

\renewcommand{\phi}{\varphi}

\newcommand{\scope}[1]{\mathrm{scope}{#1}}

\newcommand{\tw}{\operatorname{tw}}

\DeclareMathOperator{\Lap}{Lap}

\newcommand{\items}{I}

\newcommand{\size}{\mathrm{size}}
\DeclareMathOperator{\local}{Score}
\newcommand{\epsstep}{\varepsilon_{\mathrm{step}}}
\renewcommand{\Pr}{\mathbb{P}}
\newcommand{\forget}{\mathrm{Forget}}
\newcommand{\fforget}{\mathrm{\footnotesize Forget}}
\DeclareMathOperator{\poly}{poly}
\newcommand{\domain}{\mathcal{D}}
\newcommand{\family}{\mathcal{F}}
\newcommand{\dataset}{X}
\newcommand{\Paren}[1]{\left(#1\right)}
\newcommand{\lo}{\mathrm{lo}}
\newcommand{\hi}{\mathrm{hi}}

\renewcommand{\setminus}{\smallsetminus}

\allowdisplaybreaks

\title{Fixed-Parameter Tractability of Private Synthetic Data Generation}
\date{\today}

\author[1]{Badih Ghazi}
\author[2]{Crist\'obal Guzm\'an}
\author[3]{Pritish Kamath} 
\author[3]{Alexander Knop}
\author[3]{Ravi Kumar}
\author[3]{Pasin Manurangsi}
\affil[1]{Google Deepmind}
\affil[2]{Institute for Mathematical and Computational Engineering, Faculty of Mathematics and School of Engineering, Pontificia Universidad Cat\'olica de Chile}
\affil[3]{Google Research}

\date{}

\begin{document}

\maketitle

\begin{abstract}
    We study the problem of generating synthetic data under differential privacy. We establish fixed-parameter tractability (FPT) for this problem where the parameter is the treewidth of the query family's incidence graph. Our algorithms attain  optimal error rates across all regimes and are realized by two different approaches:    
  the first is based on linear programming (LP) and the FPT of the separation problem for the LP dual; the second is based on a subsampled private multiplicative weights method, where we obtain FPT for sampling from Gibbs distributions.  Both approaches are unified by a dynamic programming framework over a tree decomposition.\end{abstract}

\section{Introduction}
\label{sec:introduction}

Private synthetic data generation is a widely studied problem at the heart of private data analysis, with major practical implications \citep[see, e.g.,][]{DworkR14,Barak:2007,Blum:2008,HardtR10,Abouwd:2018,Cormode:2025,Ponomareva:2025}. The goal is to use privacy-preserving methods on a sensitive dataset to create a new version that retains key statistical properties of the original.  The advantage of private synthetic data is the flexibility of sharing it without privacy concerns for downstream applications.

Let the data domain be $\domain=\{0,1\}^d$, and consider a family $\family$ of counting queries (i.e., the queries are averages of functions $f:\domain\to \{0, 1\}$). The goal of synthetic data generation is: given a dataset $X \in \domain^*$ where $n = |X|$, produce a dataset $X' \in \domain^*$ such that $f(X) \approx f(X')$ holds, where $f(X) := \frac{1}{|X|} \sum_{x \in X} f(x)$. 
Ignoring computational constraints, the optimal error rates for generating synthetic data satifying differential privacy (DP) are known in several cases. We first consider the case of $(\eps, \delta)$-DP (i.e., approximate-DP). Here, the optimal error rates\footnote{As is standard in literature, we assume that $\eps \leq 1$ and $\delta = o\Paren{\frac{1}{n}}$ throughout. Furthermore, we consider the $\ell_\infty$-error; see \Cref{sec:preliminaries} for a formal definition.} are 
\[ 
    \textstyle
    \Theta\Big(\min\Big\{ 
        \frac{\sqrt{|\domain|\ln|\family|}}{n\eps}, 
        \frac{\sqrt{|\family|\log(1/\delta)}}{n\eps}, 
        \Big(\frac{\sqrt{\ln|\domain|}\ln|\family|\ln(1/\delta)}{n\eps}\Big)^{1/2} 
    \Big\}\Big).\footnote{Upper bounds in the first regime follow from \citep[Theorem 7.2.9]{Vadhan:2017}, for the second regime from \citep{GKM21,DK22}, and for the third regime from \citep{HardtR10}. Lower bounds can be found in \citep{Vadhan:2017}.}
\] 
For $\eps$-DP (i.e., pure-DP), the optimal error rates are still open \citep{DPorg-open-problem-optimal-query-release,Nikolov:2023}, and the best known upper bound is 
\[ \textstyle O\Big(\min\Big\{ \frac{\sqrt{|\domain|\ln|\family|}}{n\eps}, \frac{|\family|}{n\eps}, \Big(\frac{\ln|\domain|\ln|\family|}{n\eps}\Big)^{1/3} \Big\}\Big). \]
For both cases, the first regime is best for small domain size $|\domain|$, the second one is best for small query family size $|\family|$, and the third one is best when the number of datapoints $n$ is small (known as the \emph{sparse setting}).
All known algorithms that attain these rates run in time exponential in $d$. 

Unfortunately, under computational constraints, it is known that producing accurate DP synthetic data 
even for modest and natural classes of queries (e.g., 2-way marginals) is hard, and the above mentioned ``exponential in $d$'' bounds are optimal
under standard cryptographic assumptions \citep{Ullman:2011}. 
Despite this fundamental roadblock, follow-up research has identified specific settings where DP synthetic data can be efficiently produced. 
A remarkable example is the US Census TopDown algorithm \citep{Abowd:2019TopDown,Abowd:2022}, which produced DP synthetic data for 
{\em hierarchical} queries, specifically, where queries are monotone disjunctions over sets from
a laminar family. 
Other works have explored methods based on  Bayesian models \citep{ZhangPrivBayes:2017}, graphical 
models~\citep{McKenna:2019,McKenna:2021, McKenna:2022}, minimax duality~\citep{Gaboardi:2014, Vietri:2020}, 
and heuristics. It is desirable to obtain a unified approach that is capable of handling general query workloads, 
and provides a clear distinction between easy and hard instances. Quoting \citep{McKenna:2021},
\begin{quote}
\small
    {\sl In general, it would be nice to have a mechanism that can automatically adapt to an analyst-provided workload, and generate synthetic data that provides high utility on the queries and tasks in that workload. Several workload-adaptive mechanisms exist, but they are generally restricted to settings where the full high-dimensional histogram can be explicitly materialized in vector form, and are thus unable to scale to high-dimensional domains.}
\end{quote}

\subsection{Contributions}

Our main contribution is to show that the optimal error rates for DP synthetic data generation over worst-case query workloads are achievable via fixed-parameter tractable (FPT) algorithms. Our techniques are simple, broadly applicable, and both unify and extend existing methods.
\begin{itemize}\item \textsl{Structural parametrization.}  We identify the {\em incidence graph treewidth} of the query class as a key parameter for achieving FPT. 
    This parameter does not necessarily grow with the query scope size and this allows us to recover, e.g., the polynomial tractability for hierarchical queries  \citep{Abowd:2019TopDown,Abowd:2022}. The FPT results for our subsequent algorithms follow from a unified dynamic programming framework over a tree decomposition, detailed in~\Cref{sec:dynamic_program_FPT}. We further show that bounded treewidth does not only improve running time, but also provides tighter error rates than previously known; see \Cref{sec:error_tw}. \item \textsl{Small family size.} 
        When $|\family|$ is small, we propose an FPT DP synthetic data generator based on the dual 
        of a natural linear program (LP) that minimizes query error  over fractional histograms over
        $\domain$. 
        We show that optimizing this dual is FPT via the ellipsoid method~\citep{Grotschel:1981}, as the separation problem is rendered FPT by our dynamic programming framework. \item \textsl{Sparse regime.}
    When $n$ is small, we adapt the standard private multiplicative weights update (PMWU)~\citep{HardtR10} method to implicitly handle the high-dimensional distributions used for the histograms. We prove that sampling from these distributions is sufficient to maintain accuracy (\Cref{thm:implicit_pmwu_accuracy}). Moreover, in \Cref{sec:Gibbs_sampling} we show that this sampling task is FPT, by computing the partition function (by the aforementioned dynamic program) and recursively sampling  conditional probabilities from root to leaves. \end{itemize}

\subsection{Related Work}

The literature on DP synthetic data generation is vast. We summarize here the work closely related to ours. 
The noisy LP was one of the earliest proposed approaches for DP synthetic data generation \citep{Barak:2007,DworkNRRV09}; this LP perspective played a key role in the development of PMWU  \citep{HardtR10,HardtLM12}. 

The hardness of generating synthetic data for all two-way marginals \citep{Ullman:2011} has motivated alternative approaches that are either heuristic or constrain the family of distributions and/or queries, to achieve algorithmic efficiency. \cite{HardtLM12} observed that if features are partitioned and each query is supported on a distinct block, then the MWU distributions factorize, yielding a compact representation; this turns out to be a low treewidth query class. 

Other approaches are based on restricting the histogram distribution to specific parametric families. The most common one  involves  probabilistic graphical models.  For example, \cite{ZhangPrivBayes:2017} construct a degree-$k$ Bayesian network privately via the exponential mechanism, which results in running time exponential in $k$, and lacks provable accuracy guarantees.  Another class of algorithms leverage ideas from graphical models and maximum entropy estimators to design DP synthetic data without making parametric assumptions
\citep{McKenna:2019,McKenna:2021,Cai:2021}.  
For marginal queries under square loss, their algorithms admit an FPT implementation, as they are based on {\em belief propagation} \citep{Koller:2009, Murphy:2012}. 
Aside from the different notion of accuracy, their results are comparable to ours when $|\family|$ is small and consists of marginal queries. 

A general approach for efficiently generating private synthetic data in the statistical setting was proposed by \cite{Boedihardjo:2021}.  They obtain (approximately) optimal rates for the small $|\family|$ regime by solving a saddle-point problem on the histogram-query product space, similar to our approach. For computational efficiency, they assume access to a sampler with bounded R\'enyi condition number with respect to a target distribution, which is used to drastically reduce the histogram dimension.  The results are only applicable to the statistical setting, and under strong side information on the target distribution. For $k$-way marginal queries their algorithm is FPT in $k$, a running time that is not attainable in general \citep{Ullman:2011}.

A well-established technique to achieve DP for optimization, saddle-point problems (both of them relevant for DP synthetic data), and online learning, is based on combining MWU approaches with sampling from the resulting Gibbs distributions \citep[see, e.g.,][]{Hsu:2013,Gaboardi:2014,Vietri:2020,Asi:23a,Gonzalez:2026}. Crucially, in all of these works, privacy is achieved by the sampling  itself, as it corresponds to an application of the exponential mechanism. Since these solutions become accurate only when collecting sufficiently many samples, there is a degradation of the privacy budget for larger samples (which can be effectively mitigated in some cases). In our approach, privacy is guaranteed by the exponential mechanism over queries, and therefore the sampling can be performed repeatedly without degrading privacy. Our choice of sample size is only limited by the required concentration for accurate estimates of the queries, as well as the need to keep a moderate sample size for efficiency purposes.

Within the sampling approaches outlined above, some works have 
explored the possibilities of dualizing the synthetic data LP,  using
integer programming-based algorithms for the maximization subproblems. These problems are NP-hard in the worst-case but admit FPT for suitable notions of treewidth \citep{Gaboardi:2014, Vietri:2020}. Unfortunately, the resulting error rates are suboptimal compared to those obtained by PMWU. 

There is a vast FPT literature in  optimization, sampling and decision problems in connection to various notions of 
treewidth \citep[see][and the references therein]{Cygan:2015, SamerSzeider:2010, Bezakova:2016}. 
The majority of works in this area either focus on specific graph problems, or on general constraint satisfaction problems,
as well as more structural characterizations of FPT under bounded treewidth, such as \citeauthor{Courcelle:1990}'s Theorem~\citep{Courcelle:1990}.
Our approaches for FPT in optimization and sampling follow standard techniques from this literature: 
however, to make the presentation self-contained, we describe
the algorithm and running times in detail. 

\section{Preliminaries}

\label{sec:preliminaries}

Consider a data domain $\domain=\{0,1\}^d$. A dataset is a sequence $\dataset = (x_1, \ldots, x_n) \in \domain^{\ast}$, and we denote its size (i.e., number of datapoints) by $n = |\dataset|$. A dataset $\dataset$ can be represented as a \emph{histogram}, $h^\dataset\in \mathbb{R}_+^\domain$, where $h_x^\dataset=\frac1n\sum_{i=1}^n [x=x_i]$. We define a \emph{fractional histogram} as any element of the standard simplex $\Delta_\domain =\{h\in\mathbb{R}_+^\domain: \, \sum_{x\in \domain} h_x=1\}$. 

\paragraph{Differential Privacy and Synthetic Data Generation.}

A pair $X, X' \in \domain^{\ast}$ of datasets  are {\em neighbors} (denoted by $X \simeq X'$) if they differ by at most one substitution of their datapoints. We say that an algorithm ${\cal A}:\domain^{\ast}\to {\cal O}$ is \emph{$(\varepsilon,\delta)$-DP}  
if for every pair $X\simeq X'$, 
and for every (measurable) 
$E \subseteq {\cal O}$, $\mathbb{P}[{\cal A}(X)\in E]\leq e^{\varepsilon} \cdot \mathbb{P}[{\cal A}(X')\in E]+\delta$. When $\delta = 0$, we say that ${\cal A}$ is $\eps$-DP (i.e., pure-DP).  Properties and examples 
of DP algorithms are included in \Cref{app:basics_DP}.   

The \emph{synthetic data generation} problem is to design a (pure- or) approximate-DP algorithm ${\cal A}: \domain^{\ast}\to\domain^{\ast}$ that, on an input dataset $X$, outputs a synthetic dataset $\hat{X} = {\cal A}(X)$. The error of the synthetic dataset is evaluated on a family ${\cal F}$ of Boolean queries as 
$
\|(f(X) - f(\hat{X}))_{f\in \family}\|_{\infty}=\max_{f\in \family} |f(X)-f(\hat{X}))|. $
We say that $\cA$ is \emph{$\alpha$-accurate} if its expected error is at most $\alpha$.

\paragraph{Boolean Queries and Scope.}  
Consider a class $\family$ of Boolean queries comprised of functions $f:\domain\to\{0,1\}$. 
For such a function, we denote its average over the dataset $\dataset=(x_1,\ldots,x_n)$ as
$f(\dataset):=\frac{1}{n}\sum_{i=1}^n f(x_i)$, and for a (fractional) histogram $h\in\Delta_{\domain}$, 
we denote its expected value as $\langle f,h\rangle=\sum_{x \in \domain} f(x) \cdot h_x=\mathbb{E}_{x\sim h}[f(x)]$. 
In our context, it is often more convenient to compute a histogram rather than a dataset. Nevertheless, it is simple to turn the former into the latter via sampling, using a direct application of the Hoeffding bound and a union bound.  Formally, we have:\begin{lemma} \label{lem:subsampling}
Let $h \in \Delta_{\domain}$ be any fractional histogram, and let $X =(x_1,\ldots,x_m)\stackrel{\mathrm{i.i.d.}}{\sim} h$. Then,
\[
    \textstyle\Pr\Big[\forall f\in {\cal F}:\,\textstyle |f(X)-\langle f,h\rangle|
    \leq 
    \sqrt{\frac{1}{m}\ln \frac{2|\family|}{\beta}}\Big] \geq 1 - \beta.
\]
\end{lemma}

We consider Boolean functions $f:\domain \to\{0,1\}$ represented in the \emph{(reduced) ordered binary decision diagram (rOBDD)} form\footnote{
A \emph{BDD} is a representation of a Boolean function as a rooted directed acyclic
graph in which (i) the terminal nodes correspond to the output values 0 and 1, (ii) the internal nodes correspond to the 
variables of the Boolean function, and (iii) each internal node has two
outgoing edges, one taken if the variable is false and the other taken if the variable is true.  
An \emph{OBDD} is a BDD in which the variables appear in the same order on all paths from the root to the terminal
nodes.  An OBDD is \emph{reduced} if it contains neither redundant nodes nor
isomorphic subgraphs.  
}
\citep{Bryant:1992}. 
For our examples of interest, this representation can be produced with $O(d)$ bits;
in particular, this encoding is efficient. 
We denote $\size(f)$ as the size of the provided encoding, and for a set $\family$ of Boolean functions, we denote $\size(\family)=\max\{\size(f):f\in \family\}$. For  more details of rOBDD and examples, see \Cref{app:boolean_rOBDD}.
\begin{definition}
    Given $f:\domain \to\{0,1\}$, we define its scope as
    $\scope[f]\triangleq\{i\in[d]:\, \exists x\in \domain, f(x)\neq f(x^{\oplus i})\},$ 
    where $x^{\oplus i}$ is the vector $x$ with its $i$th coordinate flipped.
\end{definition}
The importance of this definition is that $f$ is uniquely determined given a
partial evaluation of the Boolean variables within its scope. Hence, we
adopt a slight abuse of notation: for any index set $S \subseteq [d]$ such that
$\scope[f]\subseteq S$ and any vector $y \in \{0,1\}^S$, we use $f(y)$ to
denote the unique value of $f(x)$ for any input $x$ satisfying
$x|_{\scope[f]} = y|_{\scope[f]}$.

Given a function $f$ in an rOBDD representation, one can compute its scope in $O(\size(f))$ time. This follows from the next result whose proof is straightforward and hence omitted. \begin{lemma} \label{lem:scope_computation}
For $f:\domain \to\{0,1\}$, $\scope[f]=\{i\in [d]: \mbox{$i$th variable appears in the rOBDD of $f$}\}$. 
\end{lemma}

\paragraph{Incidence Graph and Treewidth.}
For a graph $G=(V,E)$, let $V(G)=V$ and $E(G)=E$.  We denote a bipartite graph by $G = (U, V, E)$, where $U,V$ are a partition of $V$ and $E \subseteq U \times V$.  For $A\subseteq V(G)$, we define its boundary $\partial(A)=\{v\in V(G) \setminus A: \exists u\in A, \{u,v\}\in E(G)\}$.

\begin{definition}[Incidence Graph]
Let $\family$ be a family of functions $f:\domain\to\{0,1\}$. The \emph{incidence graph}  of $\family$ is the bipartite graph $G(\family)=([d], \family, \{\{i,f\} : i\in \scope[f]\})$.
\end{definition}
By \Cref{lem:scope_computation}, given  ${\cal F}$ there exists an $O(|\family|\cdot\size(\family))$ algorithm to compute its incidence graph.

\begin{definition}[Treewidth]
    Let $G = (V, E)$ be an undirected graph.  
    A \emph{tree decomposition} of $G$ is a pair $(T, \{X_t\}_{t \in V(T)})$, 
    where $T$ is a tree and each bag $X_t \subseteq V(G)$ is such that
    (i) $\cup_t X_t = V(G)$, 
    (ii) for every $\{u, v\} \in E(G)$, there is an $X_t$ such that $u, v \in X_t$, and 
    (iii) for every $u \in V(G)$, the set of nodes $\{ t \in V(T) \mid u \in X_t \}$ 
        is a connected subtree in $T$.  
    The \emph{width} of the tree decomposition is $\max_t |X_t|-1$.  
    The \emph{treewidth} of $G$ is the minimum width tree decomposition of $G$.  
\end{definition}

For an incidence graph $G(\family)$ of a family $\family$, let $\tw(\family)$ denote the treewidth of $G(\family)$. We defer some examples and applications of bounded treewidth graphs and function classes to \Cref{app:treewidth_examples}.

\section{Algorithm Based on Linear Programming} 
\label{sec:noisy_LP}

We begin by considering  
a synthetic data approach based on Linear Programming (LP). This approach yields FPT with optimal rates for the small family regime (i.e., $|\family|$ is small).

As input to our problem, we are given a set of noisy query estimates.  In particular, we assume that there is a DP mechanism that outputs a noisy count vector $\{\hat f(X): f\in \family\}\in \mathbb{R}^\family$ that is $\alpha$-accurate, i.e.,
$\mathbb{E}\|(\hat f(X)-f(X))_{f\in \family}\|_{\infty}\leq \alpha$. For example, the Laplace mechanism with composition
yields $\alpha\lesssim \frac{|\family|}{n\varepsilon}$ (see \Cref{app:basics_DP}), though improved bounds are known for certain families
~\citep{Li:2013,Nikolov:2013,Nikolov:2023,Lebeda:2025}.

Our goal is to construct a synthetic dataset (i.e., a fractional histogram $h$) that is consistent with these noisy counts. We formulate this as the following primal LP \eqref{eqn:LP_histogram}, where we wish to find a histogram $h$ that minimizes the maximum deviation $\alpha$ from the noisy counts: \begin{equation} \label{eqn:LP_histogram}\tag{$P$}
\min_{\alpha\geq 0,\, h\in\Delta_\domain}
\Big\{\alpha:\, \hat f(X)-\alpha \leq \sum_{x\in \domain} f(x) \cdot h_x \leq \hat f(X)+\alpha\,\, (\forall f\in \family) \Big\}.\vspace{-0.1cm}
\end{equation}
The main challenge in solving \eqref{eqn:LP_histogram} is the size of the histogram $h$, which contains $|\domain| = 2^d$ variables.  We can however leverage LP duality to arrive at a dual formulation \eqref{eqn:dual_LP_histogram} of this problem, in which there are $2|\family|$ variables and $2^{d}$ constraints. 
\begin{equation} \tag{$D$}  \label{eqn:dual_LP_histogram}
\max_{(y^+,y^-)\in\mathbb{R}_+^{2\family}}\Big\{\sum_{f\in \family}\hat f(X) \cdot (y_f^--y_f^+):\, \sum_{f\in \family} f(x) \cdot (y_f^--y_f^+)\leq 0\, (\forall x\in \domain),\, 
\sum_{f\in \family}(y_f^-+y_f^+)\leq 1 \Big\}.\end{equation}
We can solve the dual problem via the ellipsoid method. Namely, we will use the classical result that if we can check the feasibility of a solution (separation oracle) in polynomial time, then we can  optimize in polynomial time~\citep{Grotschel:1981}.  In our case, the separation oracle reduces to checking if a candidate vector $(y^+,y^-)$ violates any constraint, i.e., if 
\begin{equation} \label{eqn:dual_sep} 
\max_{x\in \{0,1\}^d} \sum_{f\in \family}f(x) \cdot (y_f^--y_f^{+}) \leq 0.\vspace{-0.1cm}
\end{equation}
In \Cref{sec:dynamic_program_FPT}, 
we show that this problem is FPT in terms of the treewidth of $\family$, and this yields:
\begin{lemma}
\label{lem:dual_LP}
The dual LP \eqref{eqn:dual_LP_histogram} is solvable in time $\poly(d,|\family|,2^{\tw(\family)},\size(\family))$. 
\end{lemma}
The FPT of this problem implies that privately generating synthetic data is also FPT.
\begin{theorem} \label{thm:DP_synth_data_noisy_LP}
If there is an $(\eps, \delta)$-DP $\alpha$-accurate algorithm that runs in time $T$ for answering queries from the family $\family$,
then there is an $(\eps, \delta)$-DP $O(\alpha)$-accurate algorithm for the synthetic data generation problem with running time $T + O\Paren{\frac{\log(|\family|/\alpha)}{\alpha^2}} + \poly(d,|\family|,2^{\tw(\family)},\size(\family))$. \end{theorem}
\begin{proof} 
First, we run the DP algorithm to obtain the noisy estimates $( \hat f(X) )_{f\in \family}$.
By the post-processing property of DP, any computation on these estimates remain $(\eps, \delta)$-DP.  
Next, we solve the dual LP \eqref{eqn:dual_LP_histogram}; let $(y_f^+,y_f^-)_{f\in \family}$ be the optimal dual solution found by the ellipsoid method with a separation oracle given by the algorithm from \Cref{lem:dual_LP}. By complementary slackness, we can find in time $\poly(d,|\family|,2^{\tw(\family)},\size(\family))$ a primal optimal solution to \eqref{eqn:LP_histogram}, call it $\tilde h$, with support of size $\leq 2|\family|$.  Since the true histogram is a feasible solution to \eqref{eqn:LP_histogram} with value $\leq \alpha$, the optimal solution $\tilde{h}$ also has value $\leq \alpha$. 
To conclude, we create a dataset by drawing $O\Paren{\frac{\log(|\family|/\alpha)}{\alpha^2}}$ samples from $\tilde{h}$. By \Cref{lem:subsampling}, this yields a dataset with expected error $O(\alpha)$. \end{proof}

\begin{remark}
    In the particular case where $\alpha \gtrsim |\family|/n$, the sampling argument at the end of the proof above can be substituted by a simple (closest integer) histogram rounding on $n\cdot \hat h$, leading to an error rate  $O(\alpha)$. This is the case, e.g.,~for answering queries by the Laplace mechanism with $\varepsilon$ being an absolute constant.
\end{remark}

\begin{remark}
    Another popular approach for solving LPs under separation oracles is to use (nonprivate) MWU methods \citep{Plotkin:1995,Arora:2012}. Given the structure of our dual LP, we can alternatively solve the dual using $O(\frac{d}{\alpha^2})$ calls to the separation oracle \eqref{eqn:dual_sep}. Depending on the parameter regime, this can yield faster algorithms. 
\end{remark}

\section{Algorithm Based on Multiplicative Weights Update} \label{sec:pmwu}

A second approach we propose is based on on the \emph{private multiplicative weights update (PMWU)} framework~ \citep{HardtR10,HardtLM12}. This approach yields FPT with optimal error rates for DP synthetic data generation in the sparse dataset regime (i.e., $n$ is small).

Once again, the main obstacle for making PMWU efficient is the cost of maintaining and updating a list of $2^{d}$ histogram variables. Our solution to this problem is surprisingly simple:  maintain these distributions \emph{implicitly}, and only subsample from them when needed.  See \Cref{alg:IPMWU} for a complete description.  Note that \textsc{IPMWU} is same as the standard PMWU except that we sample $X_t$ from the Gibbs distribution using our novel algorithm $\textsc{Sample}(z,\beta)$; see \Cref{sec:Gibbs_sampling}.

\begin{algorithm}[H]
\small

\SetKwInOut{KwParams}{Parameters}

\SetAlgoNoEnd

\caption{$\textsc{ImplicitPrivateMultiplicativeWeights}_{{\cal F}, \eps,\delta}(X)$ \ ($\textsc{IPMWU}_{{\cal F}, \eps,\delta}$)}
\label{alg:IPMWU}

\KwParams{
${\cal F}$: a class of Boolean functions described as rOBDD; $\eps,\delta$: privacy parameters
}
\KwIn{
Dataset $X\in (\{0,1\}^d)^{\ast}$
}
\KwOut{
Dataset $\hat X\in (\{0,1\}^d)^{\ast}$
}

\textbf{Initialization:} $z\gets 0\in \mathbb{R}^{\cal F}$: weights; $\eta>0$: stepsize; $\tau>0$: threshold; $m$: sample size

$\epsstep \gets \frac{\eps}{2T}$ if $\delta = 0$; otherwise, $\epsstep = \frac{\eps}{\sqrt{32T\ln(1/\delta)}}$

\For{$t=1,\ldots,T=\lceil 20d/\tau^2 \rceil + 1$}
{
$X_t \gets (\mbox{\textsc{Sample}}(z,\eta))^m$ \tcp*[f]{simulate $m$ samples from Gibbs distribution}

$(\hat f_t,s_t) \gets 
\mbox{\textsc{ExponentialMechanism}}(X,\score((\cdot,\cdot),X), \epsstep)$ \\
\qquad where for $f\in {\cal F}$, $s\in \{-1,+1\}$ and $X\in \domain^{\ast}$, $\score((f,s),X)\triangleq s[f(X_t)-f(X)]$

\If{{\em $s_t[\hat f_t(X_t)-\hat f_t(X)]+\mbox{Lap}\Paren{\nicefrac{1}{n\epsstep}}<\tau$}} {
    \Return{$\hat X=X_t$}
}
$z_{\hat f_t}\gets z_{\hat f_t}+s_t$ 
}
\Return{FAIL}
\end{algorithm}
 
The DP guarantee of \textsc{IPMWU} follows directly from the privacy guarantees of the exponential mechanism, Laplace mechanism, and the composition theorem (see \Cref{app:basics_DP}):

\begin{proposition}
    Algorithm $\textsc{IPMWU}_{{\cal F}, \eps, \delta}$ satisfies $(\eps,\delta)$-DP.
\end{proposition}

\subsection{Accuracy and Running Time of \textsc{IPMWU} (\Cref{alg:IPMWU})}

The correctness of IPMWU is based on the fact that it mirrors the classical PMWU method, with the only  difference being that the exponential mechanism is applied to a subsampled histogram at every step.  We show that, as long as the margin for all queries is approximately preserved, IPMWU provides essentially the same accuracy guarantees as PMWU.  

To generate the subsampled histogram, we use $\textsc{Sample}(z,\beta)$, described in \Cref{sec:Gibbs_sampling}. 
This subroutine generates i.i.d. samples from the Gibbs distribution 
$\Pr[Y=x] \propto \exp(-\beta\sum_{f\in \family} z_f f(x))$.
\begin{theorem} \label{thm:implicit_pmwu_accuracy}
For appropriate values of $\tau, m, \eta$ with $m, \nicefrac{1}{\tau} \leq \poly(d,|\family|,n)$, \Cref{alg:IPMWU} is $\alpha$-accurate with respect to $\family$ for
\[ \textstyle
\alpha\lesssim \Big(\frac{d}{n\eps} \ln\Paren{|\family|dn\eps} \Big)^{1/3} 
\mbox{ for pure-DP, and }
\alpha\lesssim \Big(\frac{\sqrt{d\ln(1/\delta)}}{n\eps} \ln\Paren{|\family|dn\eps} \Big)^{1/2} 
\mbox{ for approximate-DP.}
\]
\end{theorem}

\begin{proof}
Let $\zeta = \frac{1}{|\family| d n \eps}$, $C$ be a sufficiently large constant, $\alpha = 2\tau$ where
\begin{align*}
\tau = 
\begin{cases}
C \cdot \Big(\frac{d}{n\eps} \ln\Paren{|\family|dn\eps} \Big)^{1/3} & \text{ if } \delta = 0 \\
C \cdot \Big(\frac{\sqrt{d\ln(1/\delta)}}{n\eps} \ln\Paren{|\family|dn\eps} \Big)^{1/2} & \text{ otherwise.} 
\end{cases}
\end{align*}
It is simple to verify that, for sufficiently large $C$, $\tau > \frac{6}{n\epsstep}\ln\big( \frac{3T}{\zeta}\big)$. Finally, let $m = \left\lceil \frac{1}{\tau^2} \cdot \ln \frac{6|\family|T}{\zeta} \right\rceil$.

Our analysis relies on bounding the probability of three ``good'' events. First, from~\Cref{lem:subsampling} and a union bound (and our choice of $m$), the following holds with probability at least $1 - \nicefrac{\zeta}{3}$:
\begin{equation} \label{eqn:subsampling}
\forall t\in[T], f\in \family: \textstyle |f(X_t)-\langle f,h_t\rangle| < \nicefrac{\tau}{3}.
\end{equation}
Secondly, for the exponential mechanism, by \Cref{prop:exp_mech} and a union bound (and $\tau > \frac{6}{n\epsstep}\ln\big( \frac{3T}{\zeta}\big)$) implies that,  with probability at least $1 - \nicefrac{\zeta}{3}$,
\begin{equation} \label{eqn:exp_mech}
\forall t\in[T], f\in (\family\cup-\family):\, f(X_t)-f(X) \leq \hat f_t(X_t)-\hat f_t(X) + \nicefrac{\tau}{3}. 
\end{equation}
Thirdly, for the noisy threshold queries, we use the concentration of the Laplace distribution. Let $\xi_t\sim \Lap(1/\epsstep)$ be the noise added at round $t$. By a union bound, with probability at least $1 - \nicefrac{\zeta}{3}$,
\begin{equation} \label{eqn:conc_Lap} 
\forall t\in[T]: |\xi_t|\leq \nicefrac{\tau}{3}. 
\end{equation}
Hence, with probability at least $1 - \zeta$, all three good events occur; for the rest of the proof, we assume this holds.  We first study the case where one of the ``if'' conditions in \Cref{alg:IPMWU} is met. If this happens at round $t\in[T]$, by \eqref{eqn:exp_mech} and \eqref{eqn:conc_Lap}, then for all $f\in\family$, $|f(X_t)-f(X)| < 5\tau/3$. 

Next, we argue that the algorithm never returns FAIL. Suppose for the sake of contradiction that this occurs. Then, by \eqref{eqn:subsampling} and \eqref{eqn:conc_Lap}, all queries $(\hat f_t)_{t\in[T]}$ have margin $\tau/3$ (with respect to the Gibbs distribution) at each iteration. Thus, by a standard analysis of MWU \citep{Mohri:2018}\footnote{See also \cite[Claim 3.3]{Ghazi:2025}; our MWU analysis is written in a similar manner to theirs.}, after at most $T - 1 = \Theta(d/\tau^2)$ iterations, the Gibbs distribution must be $(\tau/3)$-accurate. As such, we must satisfy the ``if'' condition in the next iteration, a contradiction.

To conclude, the output is $(5\tau/3)$-accurate with probability at least $1 - \zeta$. Thus, the expected error of the algorithm is at most $(5\tau/3) + \zeta \leq \alpha$ as desired.
\end{proof}

To analyze the efficiency of \Cref{alg:IPMWU}, note that the exponential mechanism can be implemented in time $O(n\cdot|\family|\cdot\size(\family))$, if $\family$ is provided in rOBDD form; see \Cref{app:basics_DP}.
Therefore, the running time of \Cref{alg:IPMWU} is dominated by that of the sampling procedure. In \Cref{sec:dynamic_program_FPT} and \Cref{sec:Gibbs_sampling}, we prove the existence of an FPT sampler, which yields the following: 
\begin{theorem} \label{thm:IPMWU_run_time_accuracy}
For appropriate values of $\tau, m, \eta$,
\Cref{alg:IPMWU} is $(\eps,\delta)$-DP, $\alpha$-accurate for $\alpha$ as in \Cref{thm:implicit_pmwu_accuracy} and runs in time $\poly(d,|\family|,n,2^{\tw(\family)}, \size(\family))$. \end{theorem}

\section{Improved Error Bounds for  Bounded Treewidth Query Families}
\label{sec:error_tw}

The previous results show running time improvements under bounded treewidth families. We now observe that error rates are also improvable, obtaining error bounds that are polynomial on the number of features and privacy parameters, and exponential in terms of treewidth. Indeed, one can leverage the fact that for query families with treewidth $w$ most queries have scope sizes $O(w)$, and hence those can be estimated by $O(w)$-way marginals. Combining those marginals with the remaining large scope marginals provides a compact query family with treewidth $O(w)$, that can be used to generate synthetic data.

\begin{proposition}
    There exists an $(\varepsilon,\delta)$-DP algorithm with running time \\$\poly(|\family|,n,d^{\tw(\family)},2^{\tw^2(\family)}, \size(\family))$ that produces a synthetic dataset with error $\alpha$ where 
    \[ \alpha\lesssim \frac{(6d)^{2w+1}}{n\varepsilon} 
    \mbox{ for pure-DP, and }
    \alpha\lesssim \frac{2^{2w}\sqrt{(3d)^{2w+1}\ln(1/\delta)}}{n\varepsilon} 
    \mbox{ for approximate-DP.}
    \]
\end{proposition}
Note that for the bound above we ignore two of the three error-rate regimes discussed in the introduction. We ignore the small domain regime as it is uninteresting in this context, and the sparse regime as the resulting rate is suboptimal.
\begin{proof}
    We start by pre-processing the queries. First, we recall the fact that if $\tw({\cal F})\leq w$, then it is $w$-degenerate \citep{Cygan:2015} (this means that every subgraph has a node with degree at most $w$). We now prove that the number of queries with scope larger than $2w+1$ is at most $d$. For this argument, let ${\cal F}'$ be the set of such functions, and let $G'=G({\cal F'})$ which is a subgraph of $G({\cal F})$. By degeneracy, the average degree of this subgraph is bounded by $2w$, hence
    \[ \sum_{f\in{\cal F}'}\scope[f] \leq 2w(d+|{\cal F}'|). \]
    On the other hand, this sum is at least $(2w+1)|{\cal F}'|$. Therefore, $|{\cal F}'|\leq d$.
    
    Second, for each of the remaining queries in ${\cal F}\setminus {\cal F}'$, we decompose this function as the sum of indicators of variable assignments in its scope with output 1. Notice that each of these indicators is a $(2w+1)$-way marginal, and that the union of all such marginals has cardinality at most $\binom{d}{2w+1}\cdot2^{2w+1}\leq (2d)^{2w+1}$. Call this set ${\cal F}''$. Our pre-processed set of queries is $\overline{\cal F}={\cal F}'\cup{\cal F}''$. 
    Notice further that $|\overline{\cal F}|\leq (3d)^{2w+1}$, and that $\overline{\cal F}$ has incidence graph treewidth $4w(w+1)-1$: we defer the proof of the last fact to Lemma \ref{lem:bound_tw_marginals} in \Cref{app:improved_err_tw}.
    
    Apply Laplace or Gaussian (depending on the value of $\delta$) to the queries $\overline{\cal F}$, together with the algorithm from \Cref{thm:DP_synth_data_noisy_LP}. This algorithm runs in time $\mbox{poly}(d, |\overline{\cal F}|, n,2^{w^2}, \size({\cal F}))=\mbox{poly}(n,d^{w},2^{w^2}, \size({\cal F}))$, where we used the bound on $|{\cal F}''|$. The obtained synthetic data is $\alpha$-accurate w.r.t. ${\cal F}''$, which implies $(2^{2w+1}\alpha)$-accuracy w.r.t. ${\cal F}$, by the triangle inequality and the aforementioned decomposition of queries into marginals.
\end{proof}
From this analysis, two natural questions arise: (i) is there a treewidth-based upper bound   independent of the number of variables and (ii) is there an error bound that is subexponential in  treewidth?  We now show both these questions have negative answers.  

For (i), note that such a bound is not possible for one-way marginals. This is a function family with treewidth 1 but its optimal error rates are $\frac{d}{n\varepsilon}$ (for pure-DP) and $\frac{\sqrt{d\ln(1/\delta)}}{n\varepsilon}$ (for approximate-DP) \citep{Vadhan:2017}.

For (ii), note that for the family of all Boolean functions (for which $|{\cal F}|=2^{2^d}$ and its incidence graph treewidth is $d$), the optimal error rates are  exponential. Under pure-DP, packing arguments imply a lower bound $\alpha=\Omega\Big(\frac{2^d}{n\varepsilon}\Big)$   \citep{HardtTalwar:2010};  for approximate-DP, the hereditary discrepancy method yields a lower bound $\alpha=\Omega\Big(\frac{2^{d/2}}{n\varepsilon}\Big)$ 
\citep{Muthukrishnan:2012, Nikolov:2013}
 
\section{Dynamic Programming Algorithms for Optimization and Inference}

\label{sec:dynamic_program_FPT}

Following existing approaches, we will unify the algorithms for optimization and inference over the Boolean hypercube \citep[see, e.g.,][Chapter 20]{Murphy:2012}. Given a class ${\cal F}$ of Boolean functions, we consider the problems of computing the following values:
\[
    \max_{x\in \{0,1\}^d} \sum_{f\in {\cal F}}  y_f \cdot f(x) \qquad \mbox{and} \qquad
    \sum_{x\in \{0,1\}^d} \prod_{f\in {\cal F}}\exp\Big(-\beta r_f \cdot f(x) \Big), 
\]
where $(y_f)_{f\in {\cal F}}$ is a given real-valued vector,  $(r_f)_{f \in \cal F}$ is an integer-valued vector, and $\beta>0$ is a parameter. Notice the first one corresponds to evaluating the optimal value of an optimization problem over the Boolean hypercube, whereas the second one corresponds to the evaluation of the normalizing constant for a Gibbs distribution with potential $-r_f\sum_{f\in {\cal F}}f(x)$ (known as the partition function). Both problems are intractable in general. 

If we denote the maximum operator by $\oplus$ and the sum by $\otimes$ for the first problem, and the sum by $\oplus$ and the product by $\otimes$ for the second problem, 
we can express both problems commonly as
\begin{equation} \label{eqn:compute_opt_partition}
    \bigoplus_{x\in \{0,1\}^{d}} \bigotimes_{f\in {\cal F}} \psi(z_f \cdot f(x)),
\end{equation}
where $(z_f)_{f\in{\cal F}}$ is a real-valued vector, and $\psi(x) = x$ for the optimization case, and $\psi(x)=\exp(x)$ for the inference case. We conclude by noting that the algorithm we present is able to tackle these two problems under the same template, as both objectives enjoy a {\em commutative semiring} structure (see \Cref{app:semirings}); we denote by $\zero$ the identity element for $\oplus$ and $\one$ the identity element for $\otimes$. 

The following subclass of tree decompositions is useful for the dynamic program.
\begin{definition}[Nice tree decomposition]
We say that a tree decomposition $(T,\{X_t\}_{t\in V(T)})$ is \emph{nice} if it is a rooted tree with root $r$; $X_r=\emptyset$ and $X_\ell=\emptyset$ for every leaf $\ell$ of $T$; and every non-leaf node $t\in V(T)$ is one of the following three types:
\begin{itemize}[nosep]
        \item \textsl{Introduce node:} a node $t$ with exactly one child $t'$ such that $X_t=X_{t'}\cup\{v\}$ for some $v\notin X_{t'}$; in this case, we say that $v$ is \emph{introduced} at $t$.
        \item \textsl{Forget node:} a node $t$ with exactly one child $t'$ such that $X_t=X_{t'}\setminus\{v\}$ for some $v\in X_{t'}$; in this case, we say $v$ is \emph{forgotten} at $t$.
        \item \textsl{Join node:} a node $t$ with exactly two children $t_1,t_2$ such that $X_t=X_{t_1}=X_{t_2}$.
    \end{itemize}
\end{definition}
We note that the problem of computing a tree decomposition of width $O(\tw({\cal F}))$ can be solved in time $O(2^{\tw({\cal F})}(|{\cal F}|+d))$ \citep{Korhonen:2021}. Furthermore,  given a tree decomposition, one can efficiently compute a nice tree decomposition with the same width \cite[Lemma 7.4]{Cygan:2015}. Thus, we henceforth assume access to a nice tree decomposition of the instance, and describe our dynamic programming algorithm for problem \eqref{eqn:compute_opt_partition}.
 
Let $(T,\{X_t\}_{t\in V(T)})$ be a nice tree decomposition of width $w$ of the incidence graph $G=G({\cal F})$. The algorithm proceeds from the leaves to the root, computing a function that we specify below. 
For $t\in V(T)$, let $T_t$ the subtree of $T$ with root $t$ and let $X_t^{\downarrow}\triangleq\bigcup_{s\in V(T_t)} X_s$ be the union of all the bags in the subtree $T_t$. We also define the bag restriction to features, $I_t \triangleq X_t\cap [d]$, and $I_t^{\downarrow}\triangleq X_t^{\downarrow}\cap [d]$, as well as the bag restriction to functions, ${\cal F}_t\triangleq X_t\cap {\cal F}$.

At each node $t$ we have a {\em state} $\sigma=(x,u) \in \{0, 1\}^{X_t}$, where $x=(x_i:i\in I_t) \in \{0,1\}^{I_t} \mbox{ and } u=(u_f:f\in {\cal F}_t) \in \{0, 1\}^{{\cal F}_t},$  
and $x_i\in \{0,1\}$ denotes the assignment of Boolean variable $x_i$, whereas $u_f\in \{0,1\}$ denotes a possible value for $f$, based on the partial assignment $x$.

One should think of a state $\sigma$ as a restriction of $(y, f(y)) \in \{0, 1\}^{[d] \cup \family}$ on $X_t$ for some $y \in \domain$. Notice that such a restriction will not span all of $\{0, 1\}^{X_t}$; indeed some of the states $\sigma$ are impossible to obtain from any restriction. This is captured by our definition of (in)consistency below: 

\begin{definition}[Residuals and Consistency]
Let $R_f(t,x)\triangleq \{f(y):\, \exists y\in \{0,1\}^d,\, y|_{I_t}=x\}$ be the set of {\em residuals} of $f$ at $(t,x)$.  
We say that a state $\sigma = (x, u)$ at node $t$ is {\em consistent} if for all $f\in {\cal F}_t$, $u_f\in R_f(t,x)$.
\end{definition}

We observe that, for any function in a bag, any variable in its scope that has not been evaluated below must be included in the current bag. In particular,  the consistency condition can be verified for any function at a given bag (see \Cref{lem:scope_bag_function} for the formal justification).\begin{corollary} \label{cor:verify_residuals}
The residuals can be evaluated with only downward information, i.e., $R_f(t,x)= \{f(y): \exists y\in \{0,1\}^{I_t^{\downarrow}},\, y|_{I_t}=x\}.$
\end{corollary}
In summary, the residuals can be recursively and deterministically computed when traversing the tree from the leaves to the root. Each such update takes time $O(\size(f))$. 
Our recursion will aggregate these contributions throughout the tree in such a way that will prevent double counting. In order to achieve this, we maintain a table, denoted $\Potential$, and we will show inductively that it maintains a partial value leading to  \eqref{eqn:compute_opt_partition} over all partial feasible assignments, residuals and functions with unique values given the assignments. 
\begin{definition}[Forgotten Variables and Value Function] \label{def:value_function}
Let $(T,\{X_t\}_{t\in V(T)})$ be a tree decomposition of $G({\cal F})$. For $t\in V(T)$, let $C_t$ be the set of its children. If $s\in C_t$, let $\forget(s\to t) \triangleq X_s\setminus X_t$, and let $\forget_t^{\downarrow} \triangleq X_{t}^{\downarrow}\setminus X_t$. 
We define its value function $\Potential$ as follows: for $t\in V(T)$ and a consistent assignment $\sigma=(x,u)$,
\begin{equation}\label{eqn:potential_dynamic_program}
\Potential[t,\sigma] = \bigoplus_{y\in \{0,1\}^{I_t^{\downarrow}}:\, y|_{I_t}=x} \quad \bigotimes_{
f\in {\cal F}\cap \fforget_t^{\downarrow}} \psi(z_f \cdot f(y)).
\end{equation}
If $\sigma$ is inconsistent, then $\Potential[t,\sigma]=\zero$. 
\end{definition}
In particular, the value function at the (empty bag) root node evaluates \eqref{eqn:compute_opt_partition}. 
We also define the total contribution of functions that are forgotten at the current step.
\begin{definition}[Local Score]
    Let $\sigma^s=(x^s,u^s)$ be an assignment at node $s$ with parent node $t$. We define the \emph{local score} as follows:
    \[
\local(s\to t,\sigma^s) \triangleq \bigotimes \Big\{\psi(z_f \cdot u_f^s):\,\, f\in {\cal F}\cap \forget(s\to t) \Big\}.
    \]
    If ${\cal F}\cap\forget( s\to t)=\emptyset$, then $\local(s\to t,\sigma^s)=\one$.
\end{definition}
With this definition, we state the dynamic programming equation satisfied by the value function.

\begin{proposition} \label{prop:dynamic_program_consistent_assignment}
    Let $(T,\{X_t\}_{t\in V(T)}$ be a tree decomposition of $G$.  
    If $\sigma$ is consistent, then
    \begin{equation} \label{eqn:dyn_prog_rec}
    \Potential[t,\sigma] = \bigotimes_{s\in C_t} \bigoplus_{\sigma^s\in\{0,1\}^{X_s}: \sigma^s|_{X_s\cap X_t}=\sigma|_{X_s\cap X_t}}\local(s\to t,\sigma^s)\otimes \Potential[s,\sigma^s]. 
    \end{equation}
    Moreover, the value function \eqref{eqn:potential_dynamic_program} is uniquely determined by recursion \eqref{eqn:dyn_prog_rec}, with value $\one$ at the leaves.
\end{proposition}
The proof of this proposition is straightforward and thus omitted. \Cref{alg:dp_nice_td} contains a detailed description of the dynamic program. It is not hard to see that \Cref{alg:dp_nice_td} computes the desired values correctly, as stated below; we defer the full proof to \Cref{app:pf_dynamic_program_correctness}.

\begin{algorithm2e}[h]
\small
\caption{\textsc{Dynamic Programming for Value Function}}
\label{alg:dp_nice_td}
\DontPrintSemicolon
\SetAlgoNoEnd
\SetKwInOut{KwIn}{Input}
\SetKwInOut{KwOut}{Output}
\SetKwFunction{Cons}{Consistent}
\SetKwFunction{Res}{Restrict}
\SetKwFunction{Ext}{Extend}

\KwIn{
Boolean function class ${\cal F}$ in rOBDD form;
Nice tree decomposition $(T,\{X_t\}_{t\in V(T)})$ of  $G(\mathcal F)$
with root $r$ and $X_r=\emptyset$; 
semiring $(\oplus,\otimes,\zero,\one)$; 
$\psi$ function and coefficients $(z_f)_{f\in\mathcal F}$; 
residual sets/oracle $R_f(t,x)$.
}
\KwOut{$\Potential[r,\emptyset]$}\vspace{-0.2cm}

\BlankLine
\textbf{Notation:} $\Cons(t,\sigma)$ iff $\forall f\in\mathcal F_t,\ u_f\in R_f(t,x)$, where $\sigma = (x, u)$.\;

\For{$t$ in post-order transversal}{
  \uIf{$t$ is a leaf}{    $\Potential[t,\emptyset]\gets \one$\;
  }
  \uElseIf{$t$ is an \emph{introduce} node with unique child $t'$ and $X_t=X_{t'}\cup\{v\}$}{
    \For{state $\sigma\in\{0,1\}^{X_t}$}{
      $\Potential[t,\sigma]\gets \zero$\;
      \uIf{$\Cons(t,\sigma)$}{
        $\sigma'\gets \Res(\sigma,X_{t'})$\;
        $\Potential[t,\sigma]\gets \Potential[t',\sigma']$\;
      }
    }
  }
  \uElseIf{$t$ is a \emph{forget} node with unique child $t'$ and $\forget(t'\to t)=\{v\}$}{
    \For{state $\sigma\in\{0,1\}^{X_t}$}{
      $\Potential[t,\sigma]\gets \zero$\;
      \uIf{$\Cons(t,\sigma)$}{
        \For{$b\in\{0,1\}$}{
          $\sigma_b\gets \Ext(\sigma, v=b)$ \tcp*[f]{state over $X_{t'}$}\;
          \eIf{$v\in[d]$ }{
            $T_b\gets \Potential[t',\sigma_b]$\;
          }{$T_b\gets \local(t'\to t,\sigma_b)\ \otimes\ \Potential[t',\sigma_b]$\;
          }
          $\Potential[t,\sigma] \gets  \Potential[t,\sigma] \ \oplus\ T_b$\;
        }
      }
    }
  }
  \uElse(\tcp*[f]{join node}){$t$ is a \emph{join} node with children $t_1,t_2$ and $X_t=X_{t_1}=X_{t_2}$\;
    \For{state $\sigma\in\{0,1\}^{X_t}$}{
      $\Potential[t,\sigma]\gets \zero$\;
      \uIf{$\Cons(t,\sigma)$}{
        $\Potential[t,\sigma]\gets \Potential[t_1,\sigma]\ \otimes\ \Potential[t_2,\sigma]$\;
      }
    }
  }
}

\end{algorithm2e}
 
\begin{theorem} \label{thm:correctness_dyn_prog}
    \Cref{alg:dp_nice_td} computes the value function \eqref{eqn:potential_dynamic_program}.
\end{theorem}

Finally, we analyze the running time of the dynamic program.  Note that it depends on the number of different states at each node, which is $2^{w+1}$ for a nice tree decomposition of width $w$. 
\begin{theorem} \label{thm:run_time_dynamic_program}
The running time of \Cref{alg:dp_nice_td} is $2^{w}\poly(|{\cal F}|,d,\size({\cal F}),w)$.
\end{theorem}
\begin{proof}
    The algorithm performs $|V(T)|= O(w\cdot |V(G)|)=O(w\cdot(d+|{\cal F|}))$ update steps \cite[Lemma 7.4]{Cygan:2015}. 
    Each update could correspond to an initialization, introduction, forget, or join step. Each step takes $O(2^w\cdot w\cdot\size({\cal F}))$ time.
\end{proof}

\section{Perfect Gibbs Sampling from the Partition Function} 

\label{sec:Gibbs_sampling}

We now show how the partition function computed in the previous section provides a perfect Gibbs sampler, i.e.,~a random variable $Y$ such that
\begin{equation} \label{eqn:Gibbs} \Pr[Y=x]=\frac{1}{Z(\beta)}\prod_{f\in {\cal F}}\exp(-\beta r_f \cdot f(x)) \quad(\forall x\in \{0,1\}^{d}), 
\end{equation}
where $Z(\beta)>0$ is the partition function. 
We leverage the fact that the dynamic program that computes the partition function can store in memory the local scores, $\local(t'\to t,\sigma)$, and the value function $\Potential[t,\sigma]$. Since at every node there are at most $2^{w+1}$ states, the total storage cost of these quantities is $O(|V(T)|\cdot 2^{w+1})=O(2^w \cdot w \cdot(d+|{\cal F}|))$. 

While we are interested in sampling $Y\in \{0,1\}^d$, our algorithm samples full (consistent) assignments $\Sigma=(Y,U)\in\{0,1\}^{d}\times\{0,1\}^{{\cal F}}$. In this regard, it should be noted that for consistent full assignments, there is a bijection between the two (namely, $Y \leftrightarrow (Y, f(U))$), hence both samplers are equivalent up to marginalization. The following formula is key for sampling.

\begin{lemma} \label{lem:cond_proba}
Let $t'\in C_t$, and $\sigma\in\{0,1\}^{X_t}$, $\sigma'\in\{0,1\}^{X_{t'}}$ that coincide over $X_t\cap X_{t'}$. Then 
\[ \Pr\big[\Sigma|_{X_{t'}}=\sigma'\, : \, \Sigma|_{X_t}=\sigma \big]= \frac{\local(t'\to t,\sigma')\cdot \Potential[t',\sigma']}{\sum_{\omega\in \{0,1\}^{X_{t'}}: \omega|_{X_t\cap X_{t'}}=\sigma|_{X_t\cap X_{t'}}} \local(t'\to t,\omega)\cdot \Potential[t',\omega]}.
\]
\end{lemma}
With this formula, it is straightforward to implement the Gibbs sampler with a single pass over the tree decomposition (with previously computed table values for the partition function). Namely, starting from the root we sample the conditional probabilities of assignments at each child, based on \Cref{lem:cond_proba}. Performing this sampling over the tree leads to an assignment $\Sigma=(Y,U)\in \{0,1\}^{d+|{\cal F}|}$ whose $Y$ marginal probability is exactly given by the Gibbs distribution \eqref{eqn:Gibbs}. Each sampling step requires a discrete sampler over a set with at most $w+1$ Boolean variables, and therefore has cardinality at most $2^{w+1}$. 
\\

\begin{algorithm2e}[H]
\small
\caption{\textsc{Sample}$(z,\beta)$}
\label{alg:perfect_gibbs_sampler}
\DontPrintSemicolon
\SetAlgoNoEnd
\SetKwInOut{KwIn}{Input}
\SetKwInOut{KwOut}{Output}
\SetKwFunction{Children}{child}
\SetKwFunction{Compat}{Compat}
\SetKwFunction{Sample}{SampleDiscrete}
\SetKwFunction{WriteNew}{WriteNew}
\SetKwFunction{DFS}{DFS}

\KwIn{
$z\in\mathbb{R}^{\cal F}$: multipliers; $\beta$: temperature parameter; 
Nice tree decomposition $(T,\{X_t\}_{t\in V(T)})$ rooted at $r$ with $X_r=\emptyset$;\;
DP tables $\Potential[t,\sigma]$ and local scores $\local(t'\to t,\sigma')$ (sum--product case).
}
\KwOut{A sample $Y\in\{0,1\}^d$ from the Gibbs distribution \eqref{eqn:Gibbs}.}

\BlankLine
\textbf{Notation:} $\Sample(\{P(\omega)\})$ samples $\omega$ with $\Pr[\omega]\propto P(\omega)$.\;

\BlankLine
Initialize global assignment arrays $Y[i]\gets \bot$ for $i\in[d]$, and $U[f]\gets \bot$ for $f\in\mathcal F$\;
Let $\sigma_r\gets \emptyset$\;
\DFS{$(r,\sigma_r)$}\;
\Return{$Y$}\;

\BlankLine
\SetKwProg{Fn}{Function}{:}{}
\Fn{\DFS{$(t,\sigma_t)$}}{
  \ForEach{$t'\in \Children(t)$}{
    \tcp{Enumerate child-bag states compatible with $\sigma_t$}
    $\Omega \gets \{\omega\in\{0,1\}^{X_{t'}}:\ \sigma|_{X_t\cap X_t'}=\omega|_{X_t\cap X_{t'}}\}$\; \ForEach{$\omega\in\Omega$}{
      $P(\omega)\gets \local(t'\to t,\omega)\cdot \Potential[t',\omega]$\;
    }
    $\sigma_{t'} \gets \Sample(\{P(\omega)\}_{\omega\in\Omega})$\;

    \tcp{Record newly seen coordinates (bags guarantee consistency)}
    $\WriteNew(\sigma_{t'})$\;

    \DFS{$(t',\sigma_{t'})$}\;
  }
}

\BlankLine
\Fn{\WriteNew{$(\sigma_t)$}}{
  \tcp{Interpret $\sigma_t=(x,u)$ with $x$ on $I_t=X_t\cap[d]$ and $u$ on $\mathcal F_t=X_t\cap\mathcal F$}
  \ForEach{$i\in I_t$}{
    \If{$Y[i]=\bot$}{ $Y[i]\gets (\sigma_t)_i$ }
  }
  \ForEach{$f\in \mathcal F_f$}{
    \If{$U[f]=\bot$}{ $U[f]\gets (\sigma_t)_f$ }
  }
}
\end{algorithm2e} 
\begin{theorem} \label{thm:Gibbs_sampler_correctness_runtime}
 $\textsc{Sample}(z,\beta)$ (Algorithm \ref{alg:perfect_gibbs_sampler}) is an exact Gibbs sampler from the distribution \eqref{eqn:Gibbs}, and runs in time $2^w\poly(|{\cal F}|,d,\size(f),w)$.
\end{theorem}

\begin{proof} 

Let 
$U_t=X_t\cup X_{t'}$, hence 
\begin{equation}
\label{eqn:total_proba} 
\Pr[\Sigma|_{X_{t'}}=\sigma'\, : \, \Sigma|_{X_t}=\sigma] = \frac{\Pr[\Sigma|_{U_t}=\sigma\sqcup \sigma']}{\sum_{\omega\in \{0,1\}^{X_{t'}} : \omega|_{X_t\cap X_{t'}}=\sigma|_{X_t\cap X_{t'}}} \Pr[\Sigma|_{U_t}=\sigma\sqcup \omega] }, 
\end{equation}
where $\sigma \sqcup \omega\in\{0,1\}^{U_t}$ denotes the merge of a pair of consistent assignments $\sigma\in \{0,1\}^{X_t}$, $\omega\in\{0,1\}^{X_{t'}}$ over variables $U_t$ (in particular, such that $\sigma|_{X_t\cap X_{t'}}=\omega|_{X_t\cap X_{t'}}$). Next, 
\begin{align*}
&\Pr[\Sigma|_{U_t}=\sigma\sqcup \omega] = \frac{1}{Z(\beta)}\sum_{y: y|_{U_t}=\sigma\sqcup \omega} \prod_{f\in {\cal F}} \exp(-\beta r_f \cdot f(y)) \\ 
&= \frac{1}{Z(\beta)}\sum_{y: y|_{U_t}=\sigma\sqcup \omega} \local(t'\to t,\omega)\cdot \prod_{f\in \fforget_{t'}^{\downarrow}} \exp(-\beta r_f \cdot f(y)) \cdot  \prod_{f\notin \fforget_t^{\downarrow}}\exp(-\beta r_f \cdot f(y)).
\end{align*}
We have split the products over $f\in {\cal F}$ in three terms. The first one is the local score, that depends only on $\omega$; the second one is over $f\in \forget_{t'}^{\downarrow}$, and by \Cref{lem:scope_bag_function} we have that $\scope[f]\subseteq I_{t'}^{\downarrow}$, and therefore it depends only on variables from $V_t^{\downarrow}$; finally,  the factors $f\notin \forget_{t'}^{\downarrow}$ are such that $\scope[f]\cap I_{t'}^{\downarrow}=\emptyset$; the proof of this fact is entirely analogous to that of \Cref{lem:scope_bag_function}. In particular, factors $f\notin \forget_{t'}^{\downarrow}$ may depend only on variables outside $V_{t}^{\downarrow}$ and $\sigma$. In particular, if we let $y=(y_1,y_2)$ where $y_1\in\{0,1\}^{I_{t'}^{\downarrow}}$ and $y_2$ are the rest of the variables, then letting $C(\sigma)\triangleq \sum_{y_2}\prod_{f\notin \fforget_{t'}^{\downarrow}} 
\exp(-\beta r_f \cdot f(y_2))$, we get
\begin{align*}
\Pr[\Sigma|_{U_t}=\sigma\sqcup\omega]
=\frac{C(\sigma)}{Z(\beta)}  \cdot  \local(t'\to t,\omega)\cdot \Big(\sum_{y_1\in\{0,1\}^{I_{t'}^{\downarrow}}}\prod_{f\in \fforget_{t'}^{\downarrow}}\exp(-\beta r_f \cdot f(y_1))\Big).
\end{align*}
The first ratio in the expression above is independent of $\omega$, and therefore can be cancelled out in the quotient \eqref{eqn:total_proba}. Finally, the summation over $y_1$ corresponds to the value function \eqref{eqn:potential_dynamic_program}, 
hence
\[  \Pr\big[\Sigma|_{X_{t'}}=\sigma'\, : \, \Sigma|_{X_t}=\sigma \big] = \frac{\local(t'\to t,\sigma') \cdot \Potential[t',\sigma']}{\sum_{\omega}\local(t'\to t,\omega) \cdot \Potential[t',\omega] }.
\qedhere
\]
\end{proof}  

\section{Some Applications}

\label{section:examples}

To illustrate the scope of our results, we now discuss a few specific use cases of query families with bounded treewidth.

\subsection{Hierarchical Queries}

Hierarchical queries are an important example in private query answering and private synthetic data generation \citep{Abowd:2019TopDown, Abowd:2022,Ghazi:2023,Dawson:2023}.  Indeed, the class of hierarchical queries has treewidth bounded by the depth of its tree representation. 

Recall that ${\cal F}\subseteq\{0,1\}^{\{0,1\}^d}$ is a class of hierarchical queries if each $f\in {\cal F}$ is a monotone disjunction over a set $S_f\subseteq [d]$, $f(x)=\bigvee_{i\in S_f} x_i$, and the set system ${\cal S}\triangleq \{S_f:\, f\in{\cal F}\}$ is laminar. Note as well that for such functions, $\scope[f]=S_f$; hence, the incidence graph of the family ${\cal F}$ coincides with the incidence graph of ${\cal S}$.\footnote{The incidence graph of a set system ${\cal S}$ over a ground set $[d]$ is a bipartite graph with bipartition given by elements and sets, and where the pair $\{i,S\}\in E$ iff $i\in S$.}  Further, we recall that every laminar family has a tree representation, where the root represents the whole set $[d]$,  each internal node represents a element of ${\cal S}$ (excluding $[d],\emptyset$), and edges of this tree are given by minimal inclusions among the represented sets. Note, in particular, that the depth of this tree is given by the longest inclusion chain among sets from the laminar family. 

For laminar families, it is folklore that their incidence graph treewidth is bounded by the depth of its tree representation; we provide a proof for completeness. 
\begin{proposition} \label{prop:treewidth_laminar}
If ${\cal S}\subseteq 2^{d}$ is a laminar family, then its incidence graph treewidth is at most the depth of its tree representation. This upper bound is tight in the worst case.
\end{proposition}
\begin{proof}
Let $T$ be the tree representation of ${\cal S}$. The tree representation of ${\cal F}$, that we will call $T'$, is built by extending $T$ with leaves $t_i$ for each element $i\in S$, and attaching it to the node $t_S$ corresponding to the minimal set $S\in {\cal S}\cup\{[d]\}$ that contains it. Next we describe the bags at each node $t\in V(T')$. First, let ${\cal C}_S\triangleq \{R\in {\cal S}: S\subseteq R\}$ be the chain of sets containing $S$. For each nonleaf node $t_S$, we let $X_{t_S}={\cal C}_S$. For the leaves $t_i$ with $i\in [d]$, we let $X_{t_i}=\{i\}\cup {\cal C}_{\{i\}}$.

Note that the constructed bags have maximum size given by the depth of $T'$, which is one
more than the depth of $T$. We prove now that the above is a valid tree decomposition of ${\cal S}$. Indeed, each $i\in [d]$ is contained in the bag of $t_i$. Next, each $S\in{\cal S}$ is contained in the bag of the node $t_S$ associated to $S$. Next, for each edge $\{i,S\}$ (i.e., $i\in S$), there exists a bag that contains both elements (the bag of the $t_i$ leaf). Finally, the set of nodes containing any $i\in [d]$ or $S\in {\cal S}$ is connected, by the laminar property.

The tightness of this bound is provided by the following instance. Consider $d=2w$, and the set system ${\cal S} \triangleq \{\{1,\ldots, w,\ldots,w+\ell\}: \ell\in[w]\}$. Notice that the incidence graph of this set system $G({\cal S})$ contains $K_{w,w}$ as a subgraph; note that $\tw(K_{w,w})=w$, and by monotonicity of the treewidth under inclusion, $\tw(G({\cal S}))\geq w$. On the other hand, the longest inclusion chain has length exactly $w$, hence the previous argument shows that the treewidth of this instance is exactly $w$.
\end{proof}

We note as well that for any laminar family on a ground set of cardinality $d$, $|\family|\leq 2d-1$, so in our complexity bounds we can replace $|\family|$ by $O(d)$.

\begin{theorem} \label{thm:noisy_LP_hierarchical_queries}
    Let $\domain=\{0,1\}^{d}$ and consider a class ${\cal F}$ of hierarchical queries, 
    each of which can be represented by a tree of depth $h$. Then there exists an $(\varepsilon,\delta)$-DP algorithm that runs in time $\mathrm{poly}(n,d,2^h)$ 
    that generates synthetic data with expected accuracy 
    \[\alpha\lesssim \frac{h\ln|{\cal F}|}{n\varepsilon} \mbox{ when $\delta=0$,} \quad \mbox{ and } \alpha \lesssim \frac{\sqrt{h\ln|\family|\ln(1/\delta)}}{n\varepsilon} \mbox{ when $\delta>0.$}\]
\end{theorem}
\begin{proof}
First, we consider the pure-DP case. By the hierarchical structure,  the $\ell_1$-sensitivity of $(f(X))_{f\in {\cal F}}$ is bounded by $2h/n$. By \Cref{prop:Laplace_mech},  applying i.i.d.~noise $\mbox{Lap}(\frac{2h}{n\varepsilon})$ to all counts leads to estimates $(\hat f(X))_{f\in {\cal F}}$ that are $\varepsilon$-DP and with expected $\ell_{\infty}$-error $O\big(\frac{h\ln|{\cal F}|}{n\varepsilon}\big)$. Hence, by \Cref{thm:DP_synth_data_noisy_LP}, there exists an $\varepsilon$-DP synthetic data generator that runs in time $O(\mbox{poly}(n,|\family|,d,2^h)$ with expected error $\alpha \lesssim\frac{h\ln|{\cal F}|}{n\varepsilon}$.

Second, we address the approximate-DP case. The $\ell_2$-sensitivity of $(f(X))_{f\in {\cal F}}$ is bounded by $\Delta=\sqrt{2h}/n$. Hence adding i.i.d. Gaussian noise with standard deviation $\sigma=2\sqrt{2h}\ln(2/\delta)/[n\varepsilon]$ is $(\varepsilon,\delta)$-DP and attains expected worst-case error $O(\sqrt{\sigma\ln|\family|})=O(\frac{\sqrt{h\ln|\family|\ln(1/\delta)}}{n\varepsilon})$. Hence, by \Cref{thm:implicit_pmwu_accuracy}, there exists a DP synthetic data generator with expected error $O(\frac{\sqrt{h\ln|\family|\ln(1/\delta)}}{n\varepsilon})$.

Finally, note that both algorithms run in time $\mbox{poly}(n,d,2^h)$.
\end{proof}

\subsection{Partitioned Marginal Queries}

\label{subsubsec:example_partitioned_sets}

We next provide an example, corresponding to a class of marginal queries, which shows the benefits of \textsc{IPMWU} for sparse regimes. 

Let $\domain=\{0,1\}^d$, and partition the features into $b=d/k$ blocks, each of them of size $k$ (w.l.o.g., $k$ divides $d$); call the blocks $B_1,\ldots,B_b$.  For each block $\ell\in[b]$, select the $2^k$ different elements $a\in\{0,1\}^{B_{\ell}}$, and consider the queries $f_a(x)= \mathbf{1}(x|_{B_{\ell}}=a)$. 
Note that $|\family|=2^k\cdot b=\frac{2^k}{k}d$, and that $\tw(\family)=k$ (this is because the incidence graph is given by disjoint copies of the complete $K_{k,2^k}$ graph, whose treewidth is $k$). Finally, since each of these functions is a monotone disjunction on a set prescribed by $k$-way marginals, we have $\size(\family)=O(k)=O(d)$.

Our algorithms applied to this family run in time $\mbox{poly}(d,|\family|,n,\tw(\family),\size(\family))=\mbox{poly}(d,n,2^k)$. \Cref{thm:DP_synth_data_noisy_LP} with Laplace mechanism yields an error rate
\[ \alpha_{\mathrm{dense}}\lesssim \frac{\sqrt{2^kd\ln(1/\delta)}}{\sqrt{k}n\varepsilon},\]
whereas \Cref{thm:IPMWU_run_time_accuracy} yields an error rate
\[ \alpha_{\mathrm{sparse}} \lesssim \Big( \frac{\sqrt{d\ln(1/\delta)}[k+\ln(d^2/k)]}{n\varepsilon} \Big)^{1/2}. \]
Hence, if $\frac{2^k}{k[k+\ln(d^2/k)]}\gtrsim \frac{n\varepsilon}{\sqrt{d\ln(1/\delta)}}$, \textsc{IPMWU} attains superior rates compared to the dense case. 

\subsection{Tree-Structured Markov Random Fields}

Let $T=(I,E)$ be a tree and consider queries given by one-way marginals plus two-way marginals along the tree.
\begin{align} 
f_i(x)&=x_i & \forall i\in I \label{eqn:MRF_queries_1}  \\
f_{e,\ell}(x)&= \ell_1(x_i)\cdot \ell_2(x_j) &\forall e=\{i,j\}\in E, \, \ell_1(x),\ell_2(x)\in\{x,1-x\}. \label{eqn:MRF_queries_2} 
\end{align}
Approximation of queries for this family are useful for simple probabilistic graphical models, such as Chow--Liu trees \citep{ChowLiu:1968}.
\begin{proposition}
The family described in  \eqref{eqn:MRF_queries_1}, \eqref{eqn:MRF_queries_2}  has constant treewidth. 
\end{proposition}
\begin{proof}
    Let $T$ be the tree corresponding to the  decomposition. For describing the bags, consider a post-order transversal of $T$. Starting from the leaves, if the child-parent pairs are $(s,t)$ and if we denote $e=\{t,s\}$, the we let $X_s=\{s,t\}\cup\{f_{e,\ell} \mid \ell_1(x),\ell_2(x)\in\{x,1-x\}\}$.  Each variable, function, and edge is covered, and the set of bags containing each of them is connected, by construction. Finally, the treewidth is constant since each bag has constant size.
\end{proof}

\subsection{Geometric and Spatial Families}

These families arise when variables represent physical locations. They exhibit treewidth proportional to the size of the boundary of the domain.

Consider the lattice $\domain=\{0,1\}^{[I_1]\times\dots\times[I_d]}$ and the queries indexed by a grid, namely, let $E=\{\{(i_1,\ldots,i_d),(i_1',\ldots,i_d')\} \mid (i_1'=i_1+1) \veebar \cdots \veebar (i_d'=i_d+1) \}$, and 
\[
f_{e,\ell}(x)=\ell_1(x_u)\cdot \ell_2(x_v),\quad (\forall e=\{u,v\}\in E)(\forall \ell_1(x),\ell_2(x)\in\{x,1-x,1\}). 
\]

\begin{proposition}
    Suppose $I_1\leq \dots\leq I_d$. Then the family of lattice queries has treewidth $O(\prod_{k=1}^{d-1}I_k)$.
\end{proposition}

\begin{proof}
    The proof follows from an inductive formula on the treewidth of a cartesian product of graphs \citep{Djelloul:2009}.
\end{proof}

There are two simple generalizations one can consider. First, in two dimensions, it is known more generally that planar graphs have treewidth bounded linearly in terms of their diameter \citep{Lipton:1979,Baker:1994, Eppstein:2000}. This bound can also be used in our setting, if queries exhibit a planar structure. Second, in arbitrary dimensions, we can consider query families beyond those given by graphs, e.g., by considering counts over balls of a given radius. Clearly, this is a generalization of grid queries, for which it is known that the treewidth scales as $O(n^{1-1/d})$ \citep{Miller:1997}. Such generalization can be potentially useful for near-neighbor search queries.

\section*{Acknowledgements}

Part of C.~Guzm\'an's work was done as a visiting researcher at Google. C.~Guzm\'an's research was
partially supported by ANID FONDECYT 1251029 grant, and National Center for Artificial Intelligence CENIA FB210017, Basal ANID. We would like to thank Jingcheng Liu, Victor Verdugo and Jos\'e Verschae for valuable discussions about this work, as well as the useful feedback from anonymous reviewers.

\bibliography{main.bbl} 

\appendix

\section{Additional Background}

\subsection{Differential Privacy}
\label{app:basics_DP}

We start with the well known Laplace mechanisn, which achieves pure-DP.

\begin{proposition}[Laplace mechanism]
    \label{prop:Laplace_mech}
    For any function $f:\domain\mapsto\mathbb{R}$ such that for any $X\simeq X'$ satisfies $|f(X)-f(X')|\leq \Delta$, the mechanism that returns $\hat f(X) =  f(X)+\mathrm{Lap}\big(\frac{\Delta}{\varepsilon}\big),$\footnote{Recall the Laplace distribution $\mbox{Lap}(b)$ has density function $p(x)=\frac{1}{2b}\exp(-|x|/b)$.} satisfies $\varepsilon$-DP. Furthermore, this mechanism enjoys the following error bounds:
    \begin{align*} 
    \mathbb{P}[|\hat f(X)-f(X)|>\tau] &= \exp(-\varepsilon\tau/\Delta) \qquad(\forall \tau>0), \\
    \mathbb{E}[|\hat f(X)-f(X)|] &\leq \frac{\Delta}{\varepsilon}.
    \end{align*}
    For high dimensions, where $F:\domain\mapsto\mathbb{R}^k$, considering the $\ell_1$-sensitivity, $\Delta=\sup_{X\simeq X'}\|F(X)-F(X')\|_1$, the mechanism that adds Laplace noise $\hat F(X)=F(X)+\mathrm{Lap}\big(\frac{\Delta}{\varepsilon}\big)^k$, and satisfies
    \begin{align*} 
    \mathbb{P}[\|\hat F(X)-F(X)\|_{\infty}<\tau] &\leq |{\cal F}|\exp(-\varepsilon \tau/\Delta),  \\
    \mathbb{E}[\|\hat F(X)-F(X)\|_{\infty}] &\leq \frac{\Delta}{\varepsilon}(1+\ln k).
    \end{align*}
\end{proposition}

For private optimization problems a standard algorithm is the {\em exponential mechanism} \citep{McSherry:2007}. We consider the setting of an optimization problem where the objective is also a function of a sensitive dataset, where we denote the (maximization) objective by $\score:{\cal Y}\times\domain^{\ast}\mapsto\mathbb{R}$. The exponential mechanism selects a decision $y\in {\cal Y}$ with probability
 \[ \mathbb{P}[Y=y] = \frac{\exp\big(\frac{\varepsilon}{\Delta}\score(y,X)\big)}{\sum_{w\in {\cal Y}}\exp\big(\frac{\varepsilon}{\Delta}\score(w,X)\big)}.\]
 We also show an efficient implementation of this sampling based on Gumbel noise $Z\sim \mbox{Gumbel}(b)$, i.e., with density function $p_b(x)=\frac{1}{b} \exp\Big(-\frac{x}{b}-e^{-x/b}\Big)$ for $x\in \mathbb{R}$ (for details see, e.g., \citep{Durfee:2019}).

\begin{algorithm}[H]
\small

\SetAlgoNoEnd

\caption{\textsc{ExponentialMechanism}$(X,\score(y,X): y\in{\cal Y},\varepsilon)$ \textsc(ExpMech)}
\label{alg:exp_mech}

\KwIn{
dataset $X\in (\{0,1\}^d)^{\ast}$; $\score$ an objective function on variable $y\in {\cal Y}$; $\varepsilon$ privacy parameter
}
\KwOut{
$Y\in {\cal Y}$
}

\For{$y\in {\cal Y}$}
{
$s(y)\gets \score(y,X)+Z_y$; $(Z_y)_{y\in {\cal Y}}\stackrel{\mathrm{i.i.d.}}{\sim}\mbox{Gumbel}(\Delta/\varepsilon)$ }
\Return{$Y\in\arg\max\{s(y):y\in {\cal Y}\}$} 
\end{algorithm}
 
\begin{proposition}{\bf \cite[Theorems 3.10 and 3.11]{DworkR14}}
 \label{prop:exp_mech}
 Let $\score:{\cal Y}\times\domain^{\ast}\mapsto\mathbb{R}$ be a function of a decision variable $y\in {\cal Y}$ and a private input $X\in \domain^{\ast}$. If 
 \[\Delta \triangleq \sup_{X\simeq X'} \max_{y\in {\cal Y}} |\score(y,X)-\score(y,X')|<+\infty,\] 
 then \Cref{alg:exp_mech}  
 is $\varepsilon$-DP, and satisfies
 \[\mathbb{P}\Big[ \score(Y,X)\leq \max_{w\in {\cal Y}}\score(w,X)-\frac{2\Delta}{\varepsilon}\Big(\ln(|{\cal Y}|)+\tau\Big) \Big]\leq \exp(-\tau) \qquad(\forall \tau>0).\]
\end{proposition}

A useful property of DP is that it is preserved under postprocessing; namely, if ${\cal A}$ is $(\varepsilon,\delta)$-DP, any (data independent) function $F:{\cal O}\mapsto{\cal O}'$ is such that $F\circ {\cal A}$ is $(\varepsilon,\delta)$-DP. Another important property is its robustness under adaptive composition.

\begin{proposition}
    The adaptive composition of $k$ mechanisms, each satisfying $(\varepsilon,\delta)$-DP, satisfies $(k\varepsilon,k\delta)$-DP. Furthermore, for any $0<\delta'\leq 1$ it also satisfies $(\varepsilon',k\delta+\delta')$-DP, where $\varepsilon'=\varepsilon\big[ \sqrt{2k\ln(1/\delta')}+k\frac{e^{\varepsilon}-1}{e^{\varepsilon}+1}\big]$.
\end{proposition}

\subsection{Boolean Functions and Reduced Ordered Binary Decision Diagrams}
\label{app:boolean_rOBDD}

\begin{figure}[h]
    \centering\footnotesize
    \begin{tikzpicture}[
bddnode/.style={
            circle, 
            draw=black, 
            very thick, 
            minimum size=6mm, 
            inner sep=1pt, 
            fill=white,
            font=\bfseries
        },
        terminal/.style={
            rectangle, 
            draw=black, 
            very thick, 
            minimum size=5mm, 
            fill=gray!10,
            font=\bfseries
        },
low/.style={
            solid, 
            ->, 
            >=latex, 
            thick, 
            color=red!70!black
        },
        high/.style={
            solid, 
            ->, 
            >=latex, 
            thick, 
            color=blue!70!black
        },
        edge_label/.style={
            midway, 
text=black,
            fill=white,
            inner sep=1pt
        },
        scale=0.7
    ]

\node[bddnode] (x1) at (0, 0) {$x_1$};
    
\node[bddnode] (x2_left)  at (-2, -2.5) {$x_2$};
    \node[bddnode] (x2_right) at (2, -2.5)  {$x_2$};
    
\node[bddnode] (x3) at (0, -5) {$x_3$};
    
\node[terminal] (T0) at (-1.5, -7.5) {0};
    \node[terminal] (T1) at (1.5, -7.5)  {1};

\draw[low]  (x1) -- (x2_left)  node[edge_label] {0};
    \draw[high] (x1) -- (x2_right) node[edge_label] {1};
    
\draw[low]  (x2_left) -- (T0) node[edge_label] {0}; \draw[high] (x2_left) -- (x3) node[edge_label] {1}; 

\draw[low]  (x2_right) -- (x3) node[edge_label] {0}; \draw[high] (x2_right) -- (T1) node[edge_label] {1}; 

\draw[low]  (x3) -- (T0) node[edge_label] {0};
    \draw[high] (x3) -- (T1) node[edge_label] {1};
    
\node[anchor=east, color=gray] at (-4, 0) {Level $x_1$};
    \node[anchor=east, color=gray] at (-4, -2.5) {Level $x_2$};
    \node[anchor=east, color=gray] at (-4, -5) {Level $x_3$};
    \node[anchor=east, color=gray] at (-4, -7.5) {Terminals};

    \node[anchor=north west] (table) at (4.5, 0.5) {
        \begin{tabular}{ccc|c}
            \multicolumn{4}{c}{\textbf{Truth Table of $\bm f$}} \\
\toprule
            $x_1$ & $x_2$ & $x_3$ & $f$ \\
            \midrule
            0 & 0 & 0 & \textbf{0} \\
            0 & 0 & 1 & \textbf{0} \\
            0 & 1 & 0 & \textbf{0} \\
            0 & 1 & 1 & \textbf{1} \\
1 & 0 & 0 & \textbf{0} \\
            1 & 0 & 1 & \textbf{1} \\
            1 & 1 & 0 & \textbf{1} \\
            1 & 1 & 1 & \textbf{1} \\
            \bottomrule
        \end{tabular}
    };
    
\node[draw=black, fill=white, thin, inner sep=3pt, below=1.5mm] at (table.south){
        \begin{tikzpicture}
            \draw[low] (0.2,0.5) -- (1,0.5) node[right, black] {False (0)};
            \draw[high] (0.2,0) -- (1,0) node[right, black] {True (1)};
        \end{tikzpicture}
    };

    \end{tikzpicture}
    \caption{Reduced Ordered Binary Decision Diagram (rOBDD) for the 3-variable Majority function $f = (x_1 \land x_2) \lor (x_2 \land x_3) \lor (x_1 \land x_3)$, showing shared sub-graphs.}
    \label{fig:majority-robdd}
\end{figure}
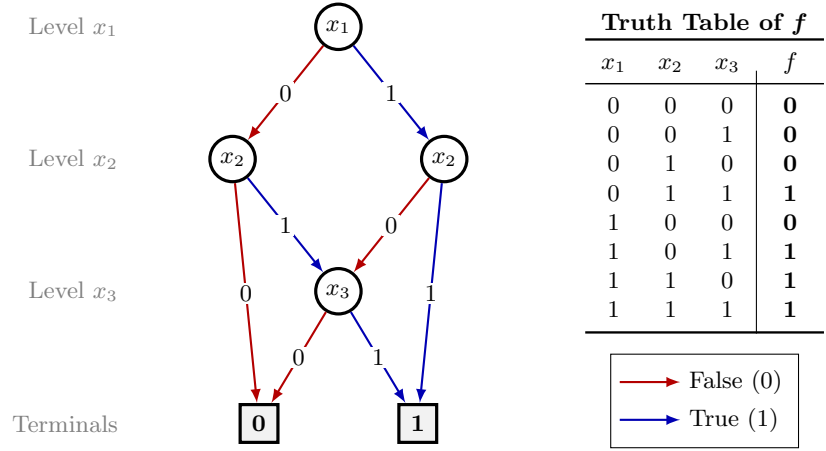

Let $\items=[d]$ be a finite set (we will refer to it as the {\em ground set}), and ${\cal F}$ a class of Boolean functions $f:\{0,1\}^{\items}\mapsto\{0,1\}$. We recall that any Boolean function can be represented as a (reduced) 
Ordered Binary Decision Diagram (OBDD) \citep{Bryant:1992}. This representation is given by a graph, which in its unreduced form is a binary tree of depth $|\items|+1$, and a permutation $\pi$ of $[|\items|]$, where all nodes $u$ at depth $i+1$ are associated to the variable $x_{\pi(i)}$, this node has two children where one of them denotes an assignment of $x_{\pi(i)}=0$ (we call this the lower child, $\lo(u)$) and the other an assignment $x_{\pi(i)}=1$ (the upper node, $\hi(u)$). Leaves have 0-1 labels, and if $x\in\{0,1\}^{\items}$, following the path indicated by its assignment, then its label is $f(x)$. This tree representation can be substantially reduced by performing the following operations:
\begin{itemize}
    \item Remove duplicate terminals: Remove all but one terminal vertex with the same label, and keep only the edge of the remaining vertex.
    \item Remove duplicate nonterminals: For nonterminal vertices $u,v$ such that $\lo(u)=\lo(v)$ and $\hi(u)=\hi(v)$, merge these nodes and keep the corresponding edges.
    \item Remove redundant tests: If a nonterminal vertex $u$ has $\lo(u)=\hi(u)$, remove this node and redirect all incoming edges to $\lo(u)$. 
\end{itemize}
See \Cref{fig:majority-robdd} for an example. 
It should be noted that, given an order of variables, the rOBDD is unique (up to isomorphism). 
This reduced representation is rather compact for many structured examples. Rather than picking a minimal representation which is hard to compute, we will consider the size of the representation given as part of the input,  hence the complexity of our algorithms will be dependent on the size. We denote this size by $\size(f)$, and $\size({\cal F})=\max\{\size(f):f\in {\cal F}\}.$ 

We note this representation is useful in our cases of interest:
\begin{enumerate}
\item {\bf Monotone disjunctions.} Here, each function $f(x)=\bigvee_{i\in S}x_i$ for some set $S\subseteq I$. The rOBDD of such function only requires depth $|S|$: w.l.o.g. $S=\{1,\ldots,k\}$ and for each $l< k$, $\lo(l)$ connects this node to $l+1$ and $\hi(l)$ directly connects to label 1 for any $l$; for $k=l$, $\lo(l)$ connects to label 0, and $\hi(l)$ to label 1. Note that $\size(f)=O(k)=O(d)$.
\item {\bf Marginal queries.} Here, a function is of the form $f(x)=\bigwedge_{i\in S} l_i(x_i)$, where $(l_i)_{i\in S}$ are either the identity function or its negation. The corresponding rOBDD has depth $|S|$ and can be similarly constructed as in the previous example (the assignments of lo and hi depend on $(l_i)_{i\in S}$). Note that $\size(f)=O(|S|)=O(d)$.
\item {\bf Range queries.} If we consider the lexicographic order over 0-1 strings of length $d$, we can consider range queries: if $a\leq b$ are elements of $\{0,1\}^d$, then we can define $f_{a,b}(x)=\mathbf{1}(a\leq x\leq b)$. The rOBDD of such function requires depth $d$, and by using pairwise comparisons one can show that $\size(f)=O(d)$.
\end{enumerate}
As mentioned earlier, the class of functions for which there exist compact rOBDD representations is broad. See \citep{Bryant:1992} for further discussions.
 
\subsection{Treewidth and Separators}

\label{app:treewidth_examples}

Tree decompositions are useful as they provide a natural mechanism for dynamic programming algorithms. To illustrate this idea, recall that a \emph{separator} of a connected graph is a set of vertices whose removal results in two connected components.

\begin{lemma} {\bf(see, e.g., \cite[Lemma 7.3]{Cygan:2015})} \label{lem:tree_decomposition}
Let $\{a,b\}$ be an edge of $T$. The forest $T \setminus \{a,b\}$ consists of two connected components $T^a$ and $T^b$; let $A=\bigcup_{t\in V(T^a)}X_t$ and $B=\bigcup_{t\in V(T^b)}X_t$. Then $\partial (A),\partial(B)\subseteq X_a\cap X_b.$ In particular, $(A,B)$ is a separation of $G$ with separator $X_a\cap X_b$.
\end{lemma}
Without loss of generality, the separators above have cardinality at most $k$, where $k$ is the width of the tree decomposition. For nice tree decompositions this bound may worsen to $k+1$.

The following result justifies the efficient verifiability of the consistency condition for a given assignment, presented in \Cref{cor:verify_residuals}.
\begin{lemma} \label{lem:scope_bag_function}
    Let $t\in T$, $f\in X_t$ and $i\in\scope[f]$. Let $B=\bigcup_{s\in V(T\setminus T_t)}X_s$. If $i\in B$ then $i\in X_t$.
\end{lemma}
\begin{proof}
    First, $\{i,f\}\in E(G({\cal F}))$. Let $A=V_{t}^{\downarrow}$. Since $i\in B$ and $f\in A$, then $\{i,f\}$ is an edge crossing the separation $(A,B)$, and therefore 
    $i\in \partial (A)$ and $f\in \partial(B)$. By \Cref{lem:tree_decomposition} applied to $t$ and its parent node, $\partial(A),\partial(B)\subseteq X_t$; in particular, $i\in X_t$.
\end{proof}

\subsection{Semirings}

\label{app:semirings}

Let $\ring$ be a set endowed with operations $\oplus$, $\otimes$, and let $\zero,\one\in \ring$. We say that $(\ring,\oplus,\otimes,\zero,\one)$ is a \emph{commutative semiring}  if: $(\ring,\oplus,\zero)$ is a commutative monoid; $(\ring,\otimes,\one)$ is a commutative monoid; $\otimes$ is distributive with respect to $\oplus$ (i.e., $a\otimes (b\oplus c)=(a\otimes b) \oplus (a\otimes c)$); $\mathbf{0}\otimes a=\mathbf{0}$ for all $a\in \mathbf{K}$.

The following are examples of commutative semirings: 
\begin{enumerate}
    \item \textsl{Max-plus Algebra:} $(\mathbb{R}\cup\{-\infty\},\mbox{max},+,-\infty,0)$, where $\mbox{max}(a,b)=\max\{a,b\}$.
    \item \textsl{Sum-product Algebra:} $([0,+\infty),+,\cdot,0,1)$, where $\cdot$ denotes multiplication.
\end{enumerate} 
\section{Missing Details from \Cref{sec:error_tw}}

\label{app:improved_err_tw}

\begin{lemma}
\label{lem:bound_tw_marginals}
Let ${\cal F}$ be a family of boolean functions on $d$ variables with incidence graph treewidth $w$, and where the scope of each function is bounded by $k$. Let $S_f=\{y\in \{0,1\}^{\scope[f]}:\, f(y) =1\}$ be the support of $f$ (resticted to its scope), and let $I_{y,f}(x)=\mathbf{1}(x^{\scope[f]}=y)$ be the corresponding indicator function.

Consider the modified query family ${\cal F}'$ where each function $f\in{\cal F}$ is substituted by 
\[ \{I_{y,f}:\, y\in S_f \}. \]
Then $\tw({\cal F}')\leq (w+1)k-1.$
\end{lemma}

\begin{proof}
Let $(T,\{X_t\}_{t\in V(T)})$ be a tree decomposition for ${\cal F}$ of treewidth $w$. We perform the following transformation on this graph:
\begin{enumerate}
\item For each $f\in {\cal F}$, select an arbitrary vertex $t\in T$ such that $f\in X_t$, and for each $y\in S_f$ append a new leaf bag $X_{(y,f)}$ to $t$.
\item We fill bag $X_{(y,f)}$ with function $I_{y,f}$ and variables $i\in\scope[f]$.
\item Recalling that the set of bag nodes $V_f$ that contain $f$ in their bags is a subtree, we will include all variables $i\in\scope[f]$ to such bags.
\item We remove all functions $f\in {\cal F}$ from all bags.
\end{enumerate}
It remains to prove that this construction defines a tree decomposition for ${\cal F}'$, and that its width is bounded by $(w+1)k$.  First, by construction, each variable and each function from ${\cal F}'$ belong to a bag. Now, each edge $(i,I_{y,f})$ in the incidence graph of ${\cal F}'$ is included in a bag; namely, in $X_{(y,f)}$, which follows by construction. Finally, the set of bags containing an indicator is a singleton, and therefore connected; and for each variable $i\in[d]$, note that in the original tree decomposition, $V_i\cap V_f\neq \emptyset$, thus by appending i to all bags from $V_f$ and to $V_{(y,f)}$ preserves connectedness. This proves that the construction is a valid tree decomposition.

To bound the width of our decomposition, note that the size of bags $X_{(y,f)}$ is precisely $\scope[f]+1\leq k+1$. For bags indexed by $t\in V(T)$, let $s$ the number of functions $f\in {\cal F}$ in $X_t$ (note that these were removed), those the new bag size is at most $w+1+s(k-1)=$ (this because for each removed function we included all variables in its scope). Noting that $s\leq w+1$, we get $\tw({\cal F}')\leq (w+1)k-1$.
\end{proof} 

\section{Missing Details from \Cref{sec:dynamic_program_FPT}}
\label{app:pf_dynamic_program_correctness}

\begin{proof}[Proof of \Cref{thm:correctness_dyn_prog}]
To prove the result, we describe the algorithm and explain why it satisfies recursion \eqref{eqn:dyn_prog_rec}.

\paragraph{Initialization.}
Recall that leaves have empty bags and thus we can initialize
 $\Potential[t,\emptyset]=\one$.

\paragraph{Node Updates.} 
We now consider the recursive step, for which we need a case analysis.

\noindent{\sl (i) Introduce node.}
Let $t'$ be the preceding child of our current node $t$. Regardless of what type of element is added to the bag, note that no element is forgotten, and therefore we only need to pass on the value to the next node: if $\sigma$ is consistent
\[ \Potential[t,\sigma] = \Potential[t',\sigma'],
\]
where $\sigma'=\sigma|_{X_{t'}}$ is the state vector $\sigma$ restricted to elements in $X_{t'}$ (this is abbreviated in the pseudocode with the $\sigma'\gets \Res(\sigma,X_{t'})$ operation); and if $\sigma$ is not consistent, $\Potential[t,\sigma]=\zero$. 

\noindent {\sl (ii) Forget node.}
Let $t\in V(T)$ and $t'$ be the preceding child of $t$. 

\begin{itemize} 
\item {\sl Forget feature:} Let $\forget(t'\to t)=\{i\}$. Since $i$ is not a function, we only need to pass on the values for the preceding states. If $\sigma$ is consistent,
\[ \Potential[t,\sigma] = \big( \Potential[t',\sigma_0]\oplus \Potential[t',\sigma_1] \big),\]
where, for $b\in \{0,1\}$, $\sigma_b$ is the $\sigma$ assignment extended by the additional assignment $x_i=b$ (this is abbreviated in the code as $\sigma_b\gets \Ext(\sigma, v=b)$).  If $\sigma$ is not consistent, then $\Potential[t,\sigma]=\zero$.

\item {\sl Forget function:} On the other hand, if $\forget(t'\to t)=\{f\}$, and if $\sigma$ is consistent
\[ \Potential[t,\sigma] = \bigoplus_{b\in \{0,1\}}\Big( \local(t'\to t,\sigma_b)\otimes \Potential[t',\sigma_b]\Big), \]
where $\sigma_b$ is the extension of $\sigma$ with the assignment $u_f=b$. If $\sigma$ is not consistent, $\Potential[t,\sigma]=\zero$.
\end{itemize}

\noindent {\sl (iii) Join node.}
In this case, we have a node $t$ with two children, $t_1,t_2$ with $X_t=X_{t_1}=X_{t_2}$. 
To update the value function, we note that aside from variables in $X_t$, there are no common variables in $t_1,t_2$ (this is due to the connectedness of the bags containing an element). In particular, the value decomposes as follows:
\[
\Potential[t,\sigma] =  \Potential[t_1,\sigma]\otimes\Potential[t_2,\sigma].
\]
Since each of these updates correspond to update \eqref{eqn:dyn_prog_rec} for the specific node updates, we have proved that our algorithm computes this recursion, which characterizes the value function \eqref{eqn:potential_dynamic_program}.
\end{proof} 

\end{document}